\documentclass[journal,9pt]{IEEEtran}
\usepackage[T1]{fontenc}
\usepackage[utf8]{inputenc}
\usepackage[english]{babel}
\usepackage{url}            % simple URL typesetting
\usepackage{booktabs}       % professional-quality tables
\usepackage{amsfonts}       % blackboard math symbols
\usepackage{nicefrac}       % compact symbols for 1/2, etc.
\usepackage{microtype}      % microtypography

\usepackage{fancybox,shadow,xcolor} 
\usepackage{setspace}
\usepackage{subfig}

\usepackage{placeins}
\usepackage{interval}
\usepackage{xcolor}
\usepackage{psfrag}
\usepackage{bm}
\usepackage{float}

\usepackage{graphicx}
\usepackage{amsmath,amssymb,amsthm}
\usepackage{mathrsfs}
\usepackage{mathdots}
\usepackage{algorithmic,algorithm}
\newtheorem{prop}{Proposition}

\newtheorem{thm}{Theorem}
\newtheorem{lem}{Lemma}
\newtheorem{cor}{Corollary}
\newtheorem{rem}{Remark}

\theoremstyle{definition}
\newtheorem{definition}{Definition}

\usepackage{interval}
\usepackage{balance}

\DeclareMathOperator*{\argmin}{arg\,min}

\DeclareMathOperator*{\rank}{rank}

\DeclareMathOperator*{\col}{col}

\DeclareMathOperator*{\set}{Set}

\renewcommand{\Re}{\mathbb{R}}

\usepackage{bbm}
\renewcommand{\paragraph}[1]{\smallskip\noindent\textbf{#1.} }

\newcommand{\BM}{\begin{bmatrix}}
\newcommand{\EM}{\end{bmatrix}}

\newcommand{\BBM}{\big[\begin{matrix}}
\newcommand{\EEM}{\end{matrix}\big]}

\newcommand{\bbm}{[\begin{matrix}}
\newcommand{\eem}{\end{matrix}]}

\title{Analysis of the least sum-of-minimums estimator for switched systems}
\author{Laurent  Bako 
\thanks{L. Bako  is with Laboratoire Amp\`{e}re (UMR CNRS 5005) -- Ecole Centrale de Lyon -- Universit\'{e} de Lyon,  69134, Ecully, France. 
E-mail: {\footnotesize \tt laurent.bako@ec-lyon.fr}. 
}
}
\begin{document}

\maketitle
\setstretch{1}

\begin{abstract}
This paper considers a particular parameter estimator for switched systems and analyzes its properties. 
The estimator in question is defined as the map from the data set to the solution set of an optimization problem where the to-be-optimized cost function is a sum of pointwise infima over a finite set of sub-functions. This is a hard nonconvex problem.  The paper  studies some fundamental properties of this problem such as uniqueness of the solution or boundedness of the estimation error regardless of computational considerations. The interest of the analysis is to lay out the main influential properties of the data on the performance of this (ideal) estimator. 
\end{abstract}


\section{Introduction}
A switched system is defined by a finite set  of dynamic systems together with a map, called the switching law,  which selects over time which system (subsystem) is activated \cite{Liberzon03-Book,Sun05-Book}. The switching law may be time-driven, event-driven or state-driven. 
Such systems can be viewed  as formal descriptions of physical phenomena taking place in, for example, power converters \cite{Lunze-Lamnabhi2009-Book}, video sequences (from segmentation perspective) \cite{Vidal08-Automatica}. Finding mathematical representations of switched systems is fundamental for the purpose of control, analysis or diagnosis. In this paper we discuss the theoretical performances/properties of a particular method for identifying a switched model from measurements.  

The problem of identifying switched systems directly from input-output data has been largely investigated in the recent literature. Examples of contributions include the works reported in \cite{Vidal03,Bako11-Automatica,Ozay12-TAC,Pillonetto16-Automatica,Goudjil16-ECC} most of which rely on numerical optimization. Some surveys of the topic can be found in \cite{Lauer19-Book,Garulli12-SYSID,Paoletti07} (see the references therein). It is fair to remark that a large number of computational methods have been proposed for estimating the parameters of switched systems. However, an important aspect that is not well understood yet is how the properties of the data quantitatively impact the performance of estimation methods operating on those data. In other words, the necessary properties of informativity of the data which favor correct estimation is still to be further investigated. 
In the current work we take a step forward in the study of such informativity properties.   
Note that so far, only a very few works have considered the fundamental question of characterizing data informativity (richness) in the context of  switched system identification \cite{Petreczky11-CDC}, \cite{Vidal08-Automatica}. \cite{Petreczky11-CDC} sketches a broad purpose condition of persistence of excitation for estimating switched state-space realizations. As to the characterization formulated in \cite{Vidal08-Automatica}, it can be interpreted as a rank condition in a lifted space (resulting from polynomial embedding of the regressors). However, neither of these contributions proposed a characterization of the parametric estimation error bound as an explicit function of the informativity degree of the regression data.

The goal of this paper is to analyze the properties of a particular estimator which we call here the Least Sum-of-Minimums  estimator (LSM) for switched system identification. This estimator maps the data  to the parameter space (of the constituent subsystems) by associating to a given data set the minimizing set of some data-dependent cost function. The cost function is formed as a sum of pointwise infima of the prediction errors associated to each subsystem. While the prediction errors may be measured in the LSM framework with multiple different loss functions, we focus specifically on the case of the absolute deviation loss function. We note that the LSM estimator is neither analytically expressible, nor numerically solvable directly at a reasonable computational price. Heuristics exist however that allow to approach the solution with, sometimes, guarantees of optimality.  For a numerical approach to this problem we refer for example, to \cite{Lauer18-Automatica}. The perspective taken here is formal rather than computational, the goal being to lay out the properties the data should enjoy to allow for an adequate retrieval of the system parameters, at least in principle.  In the wake of our previous work reported in \cite{Bako11-Automatica}, we first derive conditions on the data that guarantee exact recoverability of the true parameter matrix in the hypothetical  scenario where the measurements would be essentially noise-free. A striking property of the absolute deviation loss (used in the framework of the LSM estimator) is that it allows for exact recovery even in the face of a \textit{sparse noise}, provided that the number of nonzero values in the sparse noise sequence does not exceed a certain threshold prescribed by the informativity degree of
the data.  In the more realistic situations where the data are affected by both \textit{dense and sparse noise}, we provide parametric error bounds for the estimates delivered by the estimator. The interest of our results reside in the fact that they reveal the impact of the data informativity on the attainable performance of the (ideal) switched system estimator. This feature makes them potentially useful for optimal experiment design, that is, the  process of defining adequately the data-generating experimental conditions that would lead to the smallest (estimation) uncertainty bound.     

\paragraph{Outline}
We state the switched system identification problem in Section \ref{sec:ID} and define the LSM estimator. We start the analysis by considering essentially the noiseless scenario in Section \ref{sec:Basic-Properties}  and then the noisy one in Section \ref{sec:Error-Bounds}. The main conclusions of our study are recapitulated  in Section \ref{sec:Conclusion}.

\paragraph{Notation}
$\Re$ denotes the set of real numbers; $\Re_+$ is the set of nonnegative real numbers. 
For a matrix $A=\bbm a_1 & \cdots & a_s\eem \in \Re^{n\times s}$, we use $\set(A)$ to denote the finite set formed with the columns of $A$, i.e., $\set(A)=\left\{a_1,\ldots,a_s\right\}$. If $\mathcal{S}$ is a finite set, then $\left|\mathcal{S}\right|$ denotes the cardinality of $\mathcal{S}$. If $x\in \Re$ then $\left|x\right|$ is the absolute value of $x$. For $x=\big[\begin{matrix}x_1 & \cdots& x_n\end{matrix}\big]\in \Re^n$, $\left\|x\right\|_0$ will refer to the $\ell_0$ norm of $x$ (namely the number of nonzero entries in the vector $x$); and $\left\|x\right\|_1=\sum_{i}|x_i|$  will denote the $\ell_1$ norm of $x$. If $X\in \Re^{n\times N}$ is a matrix and $I\subset \left\{1,\ldots,N\right\}$ is a subset of the column index of $X$, then $X_I$ denotes the submatrix of $X$ formed with the columns of $X$ which are indexed by $I$.  Similarly, for a vector $v\in \Re^N$, $v_I$ refers to the subvector of $v$ consisting in the entries of $v$ indexed by $I$.

\section{The switched system identification problem}\label{sec:ID}
\subsection{The data-generating system}
Consider a (possibly nonlinear) switched system described by an equation of the form
\begin{equation}\label{eq:switched-sys}
	y_t=x_t^\top a_{\sigma(t)}^{\circ}+v_t,
\end{equation}
where $t\in \mathbb{Z}_+$ refers to discrete-time, $y_t\in \Re$ is the output of the system at time $t$, $x_t\in \Re^n$ is the regressor. As to $v_t$, it refers to potential additive noise component. $\sigma:\mathbb{Z}_+\rightarrow \mathbb{S}\triangleq \left\{1,\ldots,s\right\}$ defines a switching signal and  $a_i^{\circ}\in \Re^n$, $i\in \mathbb{S}$, denote some distinct parameter vectors. Eq. \eqref{eq:switched-sys} describes a switched system composed of $s$ dynamical subsystems each of which is activated one after another in time by the switching signal $\sigma$.

The model \eqref{eq:switched-sys} captures the situations where the regressor $x_t$ is directly observed or obtained through an intermedirary nonlinear mapping of some observable signal $z_t\in \Re^d$. We will assume that 
\begin{equation}\label{eq:varphi(zt)}
	x_t=\varphi(z_t)
\end{equation}
 where $\varphi:\Re^d\rightarrow\Re^n$ is some (known) linear or nonlinear map. Hence, depending on the choice of the mapping $\varphi$, the model \eqref{eq:switched-sys} can describe both linear and  nonlinear switched systems.

We further observe that the system represented by \eqref{eq:switched-sys} can be static, in which case $z_t$ is an unstructured  multivariate input vector, or dynamic. In this latter case, $z_t$ in \eqref{eq:varphi(zt)} may assume the form
\begin{equation}\label{eq:xt}
	z_t=\bbm y_{t-1} & \cdots& y_{t-n_a} & u_t^\top  & u_{t-1}^\top & \cdots & u_{t-n_b}^\top\eem^\top \in \Re^d
\end{equation}
with  $n_a$ and $n_b$ being some integers and $u_t\in \Re^{n_u}$  the input of the system. Note that $n_a$ can be taken equal to zero in which case $x_t$ reduces to $x_t=\bbm u_t^\top  & u_{t-1}^\top & \cdots & u_{t-n_b}^\top\eem^\top$ (hence yielding a switched nonlinear system of Finite Impulse Response type). 

\subsection{The least sum of minimums estimator}
For convenience we collect the true parameter vectors $a_i^{\circ}\in \Re^n$ from \eqref{eq:switched-sys} in a matrix $A^{\circ} = \bbm a_1^{\circ} & \cdots & a_s^{\circ}\eem\in \Re^{n\times s}$ which we call the true parameter matrix. Given a collection of $N$ data points 
\begin{equation}\label{eq:wN}
	\varpi^N=((x_1,y_1), \ldots,(x_N,y_N))
\end{equation}
generated by the switched system \eqref{eq:switched-sys}, the estimation problem of interest here is to estimate the parameter matrix $A^{\circ}$. 

The focus of this paper is this estimation problem. We  consider that the number $s$ of subsystems and the structural parameters $(n_a,n_b)$ entering the definition of $x_t$ in \eqref{eq:varphi(zt)}-\eqref{eq:xt} are known a priori.  Our goal is to design a map, called estimator, which maps the data $\varpi^N$ to the set of parameters describing the constituent subsystems of the switched system \eqref{eq:switched-sys}. To begin with the approach taken in this paper to such an estimation problem, let $\mathbb{T}$ and $\mathbb{S}$ denote the index sets of the data and the subsystems respectively, i.e., $\mathbb{T}=\left\{1, \ldots,N\right\}$ and $\mathbb{S}=\left\{1, \ldots,s\right\}$. Use the notation $\Sigma$ to denote the set of all maps $\sigma:\mathbb{T}\rightarrow\mathbb{S}$ (called here switching signals).    
%Let $\psi:\Re\rightarrow\Re_+$ denote the absolute value function, i.e., $\psi(x)=|x|$, and 
Consider the cost function  $\mathcal{J}^{\circ}:\Re^{n\times s}\times \Sigma\rightarrow\Re_+$ defined by 
$$\mathcal{J}^{\circ}(A,\sigma)=\sum_{t=1}^N\big|y_t-a_{\sigma(t)}^\top x_t\big| $$
where $A\in \Re^{n\times s}$ and $\sigma\in \Sigma$. 
Then a natural estimator of $A^{\circ}$ can be defined as the set-valued map $\Psi:{(\Re^{n}\times\Re)}^N\rightarrow \Re^{n\times s}$, 
$$\Psi(\varpi^N)=\Big\{\set(\hat{A}): \exists \hat{\sigma}\in \Sigma, (\hat{A},\hat{\sigma})\in \argmin_{A,\sigma}\mathcal{J}^{\circ}(A,\sigma)\Big\} $$
$\Psi(\varpi^N)$ is the set of all sets $\set(\hat{A})$ for all $\hat{A}\in\Re^{n\times s}$ such that $(\hat{A},\hat{\sigma})$ is a minimizer of $\mathcal{J}^{\circ}(A,\sigma)$ for some switching signal $\hat{\sigma}$. 
If we let 
\begin{equation}\label{eq:CostJ(A)}
	\mathcal{J}(A)=\sum_{t=1}^N\min_{i=1,\ldots,s}\big|y_t-a_i^\top x_t\big|
\end{equation} 
then it  can be easily shown that 
\begin{equation}\label{eq:PsiN}
	\Psi(\varpi^N)=\Big\{\set(\hat{A}): \hat{A}\in \argmin_{A}\mathcal{J}(A)\Big\}.
\end{equation}
Hence, minimizing $\mathcal{J}^{\circ}(A,\sigma)$ is equivalent to minimizing $\mathcal{J}(A)$ in  \eqref{eq:CostJ(A)}. The so defined $\Psi$ will be called the least sum-of-minimums (LSM) estimator. Because the prediction error is measured here in term of the absolute value loss function, we may also refer to $\Psi$ in the sequel as the absolute deviation LSM estimator. 
We start by observing that solving numerically any of these formulations of the switched identification problem is quite hard. The focus of this paper is not on this computational aspect but on the formal properties of the map $\Psi$. More precisely, we are interested in characterizing  conditions (on the data-generating system \eqref{eq:switched-sys} and on the properties of the data) under which $\Psi(\varpi^N)$ may contain a singleton (unique solution) or may be located at a bounded distance from the true parameter matrix $A^{\circ}$.  The primary interest of such conditions is to emphasize the main influential factors of the estimator's performance. From this perpective, we do not expect the intended properties to be necessarily numerically verifiable but to have a rather qualitative flavor which may serve for experiment design for instance. 
%%%

%%%%

% 
\section{Basic properties of the estimator}\label{sec:Basic-Properties}
We start by introducing  some definitions.  
\noindent For any matrix $A=\bbm a_1 & \cdots & a_s\eem\in \Re^{n\times s}$, let $\sigma_A\in \Sigma$ be a switching signal  satisfying
\begin{equation}\label{eq:sigmaA}
	\sigma_A(t)\in \argmin_{i\in \mathbb{S}}\big|y_t-x_t^\top a_i\big|
\end{equation}
for all $t\in \mathbb{T}$. The defining constraint \eqref{eq:sigmaA} of the switching signal $\sigma_A$ allows indeed for multiple choices of $\sigma_A(t)$ whenever $\argmin_{i\in \mathbb{S}}\big|y_t-x_t^\top a_i\big|\subset \mathbb{S}$ is not a singleton. One simple choice to solve this issue would be to set arbitrarily $\sigma_A(t)$ to be equal to the smallest element of $\argmin_{i\in \mathbb{S}}\big|y_t-x_t^\top a_i\big|$. However, for the purpose of our analysis we will define such $\sigma_A(t)$ in a more specific way. 
Consider  the index set
\begin{equation}
	I_i(A)=\left\{t\in \mathbb{T}: \sigma_A(t)=i\right\}. 
\end{equation}
 Then for all $A\in \Re^{n\times s}$, we have $I_i(A)\cap I_j(A)=\emptyset$ for $i\neq j$ and $\mathbb{T}=\cup_{i=1}^sI_i(A)$. 
For reasons that will become clear in the rest of the paper, it is desired here that $\min_{i\in \mathbb{S}}\left|I_i(A)\right|$ be as large as possible. That is, we want the partition $\left\{I_i(A)\right\}_{i\in \mathbb{S}}$ of $\mathbb{T}$ to be as  balanced as possible in term of the cardinalities of its members. Hence, it is of interest  to use the possible extra-degree of freedom offered by Eq. \eqref{eq:sigmaA}  to select $\sigma_A$ so as to maximize $\min_{i\in \mathbb{S}}\left|I_i(A)\right|$ subject to the constraint \eqref{eq:sigmaA}. 
In case the maximizing $\sigma_A$ is  still not unique, we can make it unique for a given $A$ by selecting the one which assigns to each $t$, the smallest admissible index $i\in \mathbb{S}$. To sum up, given $A\in \Re^{n\times s}$,  $\sigma_A$ can be selected uniquely by following the process described above. 

%%%%%%%%%%%%%%%%

\noindent Given $\sigma_A$, let us now define the vector $\phi(A)$ collecting the errors of the form $y_t-x_t^\top a_{\sigma_A(t)} $ for $t\in \mathbb{T}$, 
\begin{equation}\label{eq:phi(A)}
\phi(A)=\BM y_1-x_1^\top a_{\sigma_A(1)} & \cdots & y_N-x_N^\top a_{\sigma_A(N)}\EM^\top.
	%\phi(A)=\BM y_1-x_1^\top a_{\sigma_A(1)} & y_2-x_2^\top a_{\sigma_A(2)} & \cdots & y_N-x_N^\top a_{\sigma_A(N)}\EM^\top. 
\end{equation}
 Then the cost function $\mathcal{J}(A)$ in \eqref{eq:CostJ(A)} is the $\ell_1$ norm of $\phi(A)$,  
$\mathcal{J}(A)=\left\|\phi(A)\right\|_1. $
Note in passing that  $\mathcal{J}(A)$ is invariant under column permutation of the matrix $A$. This property implies that $\mathcal{J}(A)$ is indeed a function of $\set(A)$. Note that this is an intrinsic property of the multiple-regression problem. In other words, the invariance property of $\mathcal{J}(A)$  does not constitute any restriction on the switching mechanism of the to-be-identified data-generating system \eqref{eq:switched-sys}.

%%%%%%%%%%%%%%%%%%%%%%%%%%%%%%%%%%%%%%%%%%%%%%%%%%%%%%%%%

\subsection{Informativity measure of data and  exact recovery}\label{subsec:Informativity}
For any $r\in \left\{0,\ldots,N\right\}$, denote with $\mathcal{S}_r\subset\Re^N$ the set of $r$-sparse vectors in $\Re^N$, i.e., 
\begin{equation}\label{eq:Sr}
	\mathcal{S}_r=\left\{w\in \Re^N:  \left\|w\right\|_0\leq r\right\}.
\end{equation}
  For $A\in \Re^{n\times s}$,  define the distance $\delta_r(A)$  from $\phi(A)$ to the set $\mathcal{S}_r$ by   
\begin{equation}\label{eq:delta-R}
	\delta_r(A)=\inf_w\Big\{\left\|\phi(A)-w\right\|_1:w\in \mathcal{S}_r  \Big\}. 
\end{equation}
The so-defined $\delta_r(A)$ represents in fact the sum of the $N-r$ smallest entries (in absolute value) of $\phi(A)$. In particular, $\delta_0(A)=\left\|\phi(A)\right\|_1$ and $\delta_N(A)=0$.

\noindent For any subset $\mathcal{T}$ of $\mathbb{T}$, let $\phi_{\mathcal{T}}(A)$ refer to a subvector of $\phi(A)$ formed with the entries indexed by $\mathcal{T}$.
\begin{definition}[Concentration ratio]\label{def:Concentration-Ratio}
Consider the dataset $\varpi^N$ and the associated  map $\phi$ defined in \eqref{eq:phi(A)}. Let $r\in \left\{0,\ldots,N\right\}$.  We call $r$-th concentration ratio of $\phi$ on the dataset $\varpi^N$ expressed in \eqref{eq:wN}, the number defined by 
\begin{equation}\label{eq:nuR}
	\begin{aligned}
	\xi_r(\varpi^N) = 	\sup_{\substack{(A,A')\in (\Re^{n\times s})^2\\ \mathcal{T}\subset \mathbb{T}}}&\left\{\dfrac{\left\|\phi_{\mathcal{T}}(A)-\phi_{\mathcal{T}}(A')\right\|_1}{\left\|\phi(A)-\phi(A')\right\|_1}:\right. \\
			&\quad \bigg.\phi(A)\neq \phi(A'), \left|\mathcal{T}\right|\leq r \bigg\}. 
\end{aligned}
\end{equation}
\end{definition}
The supremum is taken here with respect to any pair $(A,A')\in (\Re^{n\times s})^2$ such that $\phi(A)\neq \phi(A')$ and over all subsets $\mathcal{T}$ of $\mathbb{T}$ whose cardinality does not exceed $r$. The supremum exists because it is applied to a set which is upper-bounded by $1$. 

%%%%%%%%%%%%%%%%%%%%%
\noindent We interpret the concentration ratio as a function  which measures quantitatively different levels $r$ of informativity of the data.  
For a given level $r$, the data $\varpi^N$ are all the more informative as   $\xi_r(\varpi^N)$ is small. Ideally, we would like $\xi_r(\varpi^N)$ to be as small as possible for the largest possible level $r$.

\noindent Computing numerically  $\xi_r(\varpi^N)$ would require in general solving a hard combinatorial optimization problem, the complexity of which might not be affordable in practice.  It can however be more cheaply over-approximated thanks to the direct observation that $\xi_r(\varpi^N) \leq r \xi_1(\varpi^N)$. This is because searching for $\xi_1(\varpi^N)$ instead of $\xi_r(\varpi^N)$ alleviates considerably the combinatorial nature of the problem. Note in passing that $\xi_r(\varpi^N)$ is an increasing function of $r$ and satisfies $\xi_0(\varpi^N)=0$ and  $\xi_N(\varpi^N)=1$.

%%%%%%%%%%%%%%%%%%%%%%%%%%%%%%%%%%%%%%%%%%%%%

\begin{rem}\label{rem:SingleSubsystem}
In the special case where $s=1$ (i.e., the situation where \eqref{eq:switched-sys} reduces to a single subsystem), the matrix $A$ reduces to a single vector, say $A=a\in \Re^n$. We recover the classical linear regression problem.  Then
$$\begin{aligned}
	\phi(A)&=\BM (y_1-x_1^\top a) & (y_2-x_2^\top a) & \cdots & (y_N-x_N^\top a)\EM^\top\\
	   & = \mathbf{y}-X^\top a.
\end{aligned}$$ 
where $X=\bbm x_1 & \cdots & x_N\eem\in \Re^{n\times N}$ is a matrix collecting all the regressors  $\left\{x_t\right\}_{t\in \mathbb{T}}$ generated by \eqref{eq:switched-sys} and $\mathbf{y}=\bbm y_1 & \cdots & y_N\eem$ is the vector of all output samples. 
In this case, $\xi_r(\varpi^N)$ in \eqref{eq:nuR} takes the form
\begin{equation}\label{eq:nuR2}
	\xi_r^\circ(\varpi^N) = \sup_{\substack{\eta\in \Re^{n}\\ \mathcal{T}\subset \mathbb{T}}}\Bigg\{\dfrac{\left\|X_{\mathcal{T}}^\top \eta\right\|_1}{\left\|X^\top \eta \right\|_1}: \eta\neq 0, \left|\mathcal{T}\right|\leq r \Bigg\} 
\end{equation}
where it is assumed that $\rank(X)=n$, that is, $X$ is full row rank. The notation $X_{\mathcal{T}}$ refers to the matrix formed with the columns of $X$ indexed by $\mathcal{T}$.
We observe that in this case, $\xi_r^\circ(\varpi^N)$ depends only on the regressor matrix $X$. Moreover, it can be overestimated through the solution of a convex optimization, see \cite{Bako17-TAC}.  
\end{rem}

Using the concentration ratio introduced in \eqref{eq:nuR}, we can now state a fundamental lemma for our analysis (see Lemma \ref{lem:INEQ} below, which can be viewed as a special reformulation of Lemma 4.2 in \cite{Daubechies10}). To ease the proof, we start with a preliminary technical  lemma. 

%%%%%%%%%%%%%%
\begin{lem}\label{lem:First-Lem}
Let   $r\in \left\{0,\ldots,N\right\}$ and $\mathcal{S}_r$ be defined as in \eqref{eq:Sr}. 
 Consider an arbitrary vector $v\in \Re^N$ and define\footnote{with the convention that $\mathcal{T}_r(v)=\emptyset$ for $r=0$. }  $\mathcal{T}_r(v)\subset \mathbb{T}$ to be the index set of the $r$ largest entries in absolute value of $v$.  
Then for all  $v'\in \Re^N$, 
\begin{equation}
\begin{aligned}
		\left\|v'-v\right\|_1-2&\big\|(v'-v)_{\mathcal{T}_r(v)}\big\|_1\\
		&\qquad \leq \left\|v'\right\|_1-\left\|v\right\|_1+2\inf_{w\in \mathcal{S}_r}\left\|w-v\right\|_1
\end{aligned}
\end{equation}
\end{lem}
%%%%%%%%%%%%%
\begin{proof}
See Appendix \ref{Proof-lem:First-Lem}. 
\end{proof}
%%%%
\begin{lem}\label{lem:INEQ}
Let $r\in \left\{0,\ldots,N\right\}$. Consider the dataset $\varpi^N$ as in \eqref{eq:wN} and  $\xi_r(\varpi^N)$ as defined in \eqref{eq:nuR}. If $\xi_r(\varpi^N)<1/2$, then 
\begin{equation}\label{eq:Bounding-phi(A)-phi(A')}
\begin{aligned}
		\left\|\phi(A')-\phi(A)\right\|_1\leq &\dfrac{1}{1-2\xi_r(\varpi^N)}\left(\mathcal{J}(A')-\mathcal{J}(A)+2\delta_r(A)\right) \\
		& \qquad \forall (A,A')\in \Re^{n\times s}\times \Re^{n\times s}
\end{aligned}
\end{equation}
with $\phi(A)$, $\mathcal{J}(A)$ and $\delta_r(A)$ defined in \eqref{eq:phi(A)}, \eqref{eq:CostJ(A)} and \eqref{eq:delta-R} respectively.
\end{lem}
%%%%%%%%%%%%%%%%%%%%%%%%%%%%%%%%%%%
\begin{proof}
Let $\mathcal{T}$ be a subset of $\mathbb{T}$ containing the indices of the $r$ largest entries of $\phi(A)$ in absolute value.  
We apply the result of Lemma \ref{lem:First-Lem} with $v=\phi(A)$ and $v'=\phi(A')$, which leads immediately to  
%%%%%
\begin{equation}\label{eq:Ineq1}
	\begin{aligned}
		\left\|\phi(A')-\phi(A)\right\|_1-&2\left\|\phi_{\mathcal{T}}(A')-\phi_{\mathcal{T}}(A)\right\|_1\\
		&	\leq \left\|\phi(A')\right\|_1-\left\|\phi(A)\right\|_1+2\delta_r(A)
\end{aligned}
\end{equation}
where $\delta_r(A)$ is defined as in \eqref{eq:delta-R}. 
From the definition \eqref{eq:nuR} of $\xi_r$, it can further be observed that 
$$\left\|\phi_{\mathcal{T}}(A')-\phi_{\mathcal{T}}(A)\right\|_1\leq \xi_r(\varpi^N) \left\|\phi(A')-\phi(A)\right\|_1, $$
which in turn implies that
$\left(1-2\xi_r(\varpi^N)\right)\left\|\phi(A')-\phi(A)\right\|_1$ is smaller than the left hand side term of \eqref{eq:Ineq1}. 
We therefore get
$$\begin{aligned}
	\big(1-2\xi_r(\varpi^N)\big)&\left\|\phi(A')-\phi(A)\right\|_1\\
	&\qquad \leq \left\|\phi(A')\right\|_1-\left\|\phi(A)\right\|_1+2\delta_r(A)
\end{aligned} $$ and the result follows. 
\end{proof}
%%%%%%%%%%%%%%%%%%%%%%%%%%%%%%%%%%%
%%%
\begin{rem}
In the scenario of Remark \ref{rem:SingleSubsystem}, the result of Lemma \ref{lem:INEQ} would read as
\begin{equation}
	\big\|X^\top \left(a'-a\right)\big\|_1\leq \dfrac{1}{1-2\xi_r^\circ(\varpi^N)}\left(\mathcal{J}(a')-\mathcal{J}(a)+2\delta_r(a)\right)
\end{equation}
with $\xi_r^\circ(\varpi^N)$ as in \eqref{eq:nuR2}. 
Hence if $X$ is full row rank then the left hand side constitutes a data-dependent norm on  the error $a'-a$. If we let 
$\lambda=\inf_{\left\|\eta\right\|_1=1}\left\|X^\top \eta\right\|_1$, then $\big\|a'-a\big\|_1\leq \dfrac{1}{\lambda(1-2\xi_r^\circ(\varpi^N))}\left(\mathcal{J}(a')-\mathcal{J}(a)+2\delta_r(a)\right)$.   
\end{rem}
\noindent 
By interchanging the roles of $A$ and $A'$ in the  inequality \eqref{eq:Bounding-phi(A)-phi(A')} one can obtain 
$$\left\|\phi(A)-\phi(A')\right\|_1\leq \dfrac{1}{1-2\xi_r(\varpi^N)}\left(\mathcal{J}(A)-\mathcal{J}(A')+2\delta_r(A')\right)$$
Summing this with \eqref{eq:Bounding-phi(A)-phi(A')} then yields the following inequality  
\begin{equation}
	\left\|\phi(A')-\phi(A)\right\|_1\leq \dfrac{1}{1-2\xi_r(\varpi^N)}\left(\delta_r(A')+\delta_r(A)\right). 
\end{equation}
Another immediate consequence of Lemma \ref{lem:INEQ} can be stated as follows:
\begin{lem}\label{lem:dist-to-opt}
If $\xi_r(\varpi^N)<1/2$ for some $r\in \left\{0,\ldots,N\right\}$, then for all  $\hat{A}\in \argmin_A \mathcal{J}(A)$ and for all $A\in \Re^{n\times s}$, 
\begin{equation}\label{eq:distance-to-optimal}
	\big\|\phi(A)-\phi(\hat{A})\big\|_1\leq \dfrac{2}{1-2\xi_r(\varpi^N)}\delta_r(A). %\quad \forall A\in \Re^{n\times s}
\end{equation}
Moreover, if there exists a matrix $\tilde{A}$ such that $\|\phi(\tilde{A})\|_0\leq r$ then 
$$ \begin{aligned}
	\argmin_A\mathcal{J}(A)&=\left\{A\in \Re^{n\times s}: \big\|\phi(A)\big\|_0\leq r \right\}\\
	&=\left\{A\in \Re^{n\times s}: \phi(A)=\phi(\tilde{A}) \right\}
\end{aligned}$$
\end{lem}
%%%%
\begin{proof}
By Eq. \eqref{eq:Bounding-phi(A)-phi(A')}, we have
$$ \big\|\phi(A)-\phi(\hat{A})\big\|_1\leq \dfrac{1}{1-2\xi_r(\varpi^N)}\left(\mathcal{J}(\hat{A})-\mathcal{J}(A)+2\delta_r(A)\right) $$
for all $A\in \Re^{n\times s}$.
Because $\mathcal{J}(\hat{A})-\mathcal{J}(A)\leq 0$, this yields immediately \eqref{eq:distance-to-optimal}.  
The second statement follows from the fact that if $\|\phi(\tilde{A})\|_0\leq r$, then $\delta_r(\tilde{A})=0$. Therefore,  replacing $A$ with $\tilde{A}$ in \eqref{eq:distance-to-optimal} shows that $\phi(\tilde{A})=\phi(\hat{A})$ and so, $\mathcal{J}(\tilde{A})=\mathcal{J}(\hat{A})$. Hence such an $\tilde{A}$  is necessarily in $\argmin_A\mathcal{J}(A)$. On the other hand, since $\phi(\tilde{A})=\phi(\hat{A})$, any $\hat{A} \in \argmin_A \mathcal{J}(A)$ satisfies $\|\phi(\hat{A})\|_0\leq r$ hence concluding the proof. 
\end{proof}
%%%%%%%%%%%%%%%%%%%%%%%%%%%%%%%%%%%%%%%%%%%%%%%%%%%%%%%%%%%%%%%%%%%%%%%%%%%%%%%%%%%%%%%%%%%%%%%%%%%

An interpretation of Lemma \ref{lem:dist-to-opt}  is that if the data $\varpi^N$ used to construct the map $\phi$ in \eqref{eq:phi(A)}
are generated by the switched system \eqref{eq:switched-sys} and if the data is sufficiently informative in the sense that $\xi_r(\varpi^N)<1/2$ for some $r$ and the system parameter vectors are such that $\|\phi(A^{\circ})\|_0\leq r$ over the data, with $A^{\circ}$ denoting the true parameter matrix (see Eq. \eqref{eq:switched-sys}),  then $\set(A^{\circ})\in \Psi(\varpi^N)$.  At this step, a question that needs to be discussed further is whether $\set(A^{\circ})$ may be the unique member of $\Psi(\varpi^N)$.  For this purpose we need a property  of uniform rank  on the data $X$.

\begin{definition}[An integer measure of genericity] \cite{Bako11-Automatica} 
\label{def:genericity}
Let $X\in \Re^{n\times N}$ be a data matrix satisfying $\rank(X)=n$.
The $n$-genericity index of $X$, denoted   $\nu_n(X)$,   is defined as 
the minimum integer $m$ such that any $n\times m$ submatrix of $X$  has rank $n$,  
\begin{equation}\label{eq:Nu_n_X}
\nu_n(X)=\!\min\Big\{m:\forall \: \mathcal{S}\subset \mathbb{T} \mbox{ with }\left|\mathcal{S}\right|=m, \:  \rank(X_\mathcal{S})=n\Big\}.
\end{equation}
Here, $X_\mathcal{S}$ is a matrix formed with the columns of $X$ indexed by $\mathcal{S}$. 
\end{definition}
%%%%%%%%%%
This definition implies that any submatrix of $X\in \Re^{n\times N}$ having at least $\nu_n(X)$ columns (with $n\leq \nu_n(X)\leq N$), has full row rank. The smaller $\nu_n(X)$, the more generic the regression data $X$ are said to be. According to this rough criterion, the most generic data $X$ achieve $\nu_n(X)=n$. This is typically the case  when the regressors $\left\{x_t\right\}_{t\in \mathbb{T}}$ are in \textit{general position} in $\Re^n$. Under some minimality conditions \cite{Petreczky20-IJNRC} on the data-generating system \eqref{eq:switched-sys}, if the input signal  $\left\{u_t\right\}$   is generated at random, then $\nu_n(X)=n$ with probability one.

%%%%%%%%%
Equipped with this notation and the definition of genericity index $\nu_n(X)$, we can now characterize uniqueness of the minimizer of $\mathcal{J}(A)$ based on the following lemma.

\begin{lem}\label{lem:injectivity}
Consider a dataset $\varpi^N$ of the form \eqref{eq:wN} and the notation $I_i(A)$ introduced at the beginning of Section \ref{sec:Basic-Properties}. Assume that there exists a matrix $\tilde{A}=\bbm \tilde{a}_1 & \cdots & \tilde{a}_s\eem\in \Re^{n\times s}$ with distinct columns $\tilde{a}_i$ such that 
\begin{equation}\label{eq:cond-Ni}
\min_{i\in \mathbb{S}}\big|I_i(\tilde{A})\big|\geq s\nu_n(X) %\quad  \forall  i\in \mathbb{S}
\end{equation}
on the data $\varpi^N$. 
Then the following holds:
\begin{equation}
	\forall A\in \Re^{n\times s}, \: \phi(A)=\phi(\tilde{A}) \: \Rightarrow \:  \set(A)=\set(\tilde{A}).
\end{equation}
\end{lem}
%%%%%%%%%%
\begin{proof}
Let $A$ be such that $\phi(A)=\phi(\tilde{A})$. Then for all $t\in \mathbb{T}$, $y_t-x_t^\top a_{\sigma_A(t)}=y_t-x_t^\top \tilde{a}_{\sigma_{\tilde{A}}(t)}$, which is equivalent to 
$x_t^\top(\tilde{a}_{\sigma_{\tilde{A}}(t)}-a_{\sigma_{A}(t)})=0$ for all $t\in \mathbb{T}$.  \\
 The next step of the proof is to show that for any  $i\in \mathbb{S}$ there exists $j^\star\in \mathbb{S}$ such that $I_{ij^\star}\triangleq I_i(\tilde{A})\cap I_{j^\star}(A)$ has a cardinality larger than or equal to $\nu_n(X)$. For this purpose we proceed by contradiction. Take an arbitrary $i\in \mathbb{S}$ and assume that  $\left|I_{ij}\right|<\nu_n(X)$ $\forall j\in \mathbb{S}$. 
Noting that
$$I_i(\tilde{A})=I_i(\tilde{A})\cap \mathbb{T} =I_i(\tilde{A})\cap (\cup_{j=1}^s I_j(A))=\cup_{j=1}^s I_{ij}, 
%\big(I_i(\tilde{A})\cap I_j(A) \big)
$$
we obtain
$|I_i(\tilde{A})| \leq \sum_{j=1}^s|I_{ij}|<s\nu_n(X).$
 But this constitutes a contradiction to the assumption \eqref{eq:cond-Ni}. In conclusion, for all $i\in \mathbb{S}$, there exists a $j^\star$ such that $|I_{ij^\star}|\geq \nu_n(X)$. 
Now we observe that for all $t\in I_{ij^\star}$, $x_t^\top(\tilde{a}_i-a_{j^\star})=0$ and so,  $X_{I_{ij^\star}}^\top(\tilde{a}_i-a_{j^\star})=0$. But since $|I_{ij^\star}|\geq \nu_n(X)$, we have $\rank(X_{I_{ij^\star}})=n$, which implies that  $\tilde{a}_i=a_{j^\star}$.  Since all columns of $\tilde{A}$ are distinct (no repetition), we conclude that $\tilde{A}$ and $A$ have the same columns up to a permutation which is equivalent to saying that  $\set(\tilde{A})=\set(A)$. 
\end{proof}
%%%%%%%%%%%%%%%%
It is interesting to note that in the absence of noise in \eqref{eq:switched-sys},  having the true parameter matrix $A^{\circ}$ to obey \eqref{eq:cond-Ni} is a sufficient condition for exact recovery of that matrix from the data. What this means is that if $v_t=0$ for all $t$ and if all the subsystems have been sufficiently excited in the sense that condition \eqref{eq:cond-Ni} holds for $\tilde{A}=A^{\circ}$, then 
$ \Psi(\varpi^N)=\left\{\set(A^{\circ})\right\}$. 
%%%%%%%%%%%%%%%%

The following theorem recapitulates the discussion of this section. 
\begin{thm}\label{thm:Estimator-Uniqueness}
Consider the dataset $\varpi^N$ in \eqref{eq:wN}, generated by the switched system \eqref{eq:switched-sys}. 
Assume that:
\begin{itemize}
\item $\varpi^N$ is informative enough in the sense that $\xi_r(\varpi^N)<1/2$ for some $r\in \left\{0,\ldots,N\right\}$; let then 
$$r^*(\varpi^N)=\max\left\{r: \xi_r(\varpi^N)<1/2 \right\}. $$
\item There exists a matrix $\tilde{A}\in \Re^{n\times s}$ satisfying the condition \eqref{eq:cond-Ni} and $\|\phi(\tilde{A})\|_0\leq r^*(\varpi^N)$.  
\end{itemize}
Then the estimator $\Psi$ defined in \eqref{eq:PsiN} satisfies
\begin{equation}
	\Psi(\varpi^N)=\big\{\set(\tilde{A})\big\}.
\end{equation}
\end{thm}
%%%
\begin{proof}
To begin with, note that for $r^*$ defined as in the statement of the theorem, it holds that $\delta_{r^*}(\tilde{A})=0$ (see Eq. \eqref{eq:delta-R} for the definition of $\delta_r$). Now, since the conditions of Lemma \ref{lem:dist-to-opt} are satisfied, we can apply it to infer that  if $\hat{A}\in \argmin_A\mathcal{J}(A)$, then $\phi(\hat{A})=\phi(\tilde{A})$ so that $\mathcal{J}(\tilde{A})=\min_A \mathcal{J}(A)$. % and $\set(\tilde{A})\in \Psi(\varpi^N)$. 
Conversely,  it is immediate to see that any $A'\in \Re^{n\times s}$ which satisfies $\phi(A')=\phi(\tilde{A})$ lies necessarily in $\argmin_A\mathcal{J}(A)$. Hence we can write 
$$ \argmin_A\mathcal{J}(A)=\left\{A\in \Re^{n\times s}: \phi(A)=\phi(\tilde{A}) \right\}$$
Applying Lemma \ref{lem:injectivity}, we can then write 
$$ \argmin_A\mathcal{J}(A)=\left\{A\in \Re^{n\times s}: \set(A)=\set(\tilde{A}) \right\}$$
and so, from \eqref{eq:PsiN} we see that  $\Psi(\varpi^N)=\big\{\set(\tilde{A})\big\}$. 
\end{proof}
%%%%
%%%
An interpretation of Theorem \ref{thm:Estimator-Uniqueness} is that if the data are sufficiently informative, then the set-valued estimator $\Psi(\varpi^N)$ returns only a singleton. We would of course like this singleton to coincide with the true set of parameter vectors $\left\{a_i^{\circ}\right\}_{i\in \mathbb{S}}$. For this to hold, it suffices that  the true parameter matrix $A^{\circ}$ satisfies the second condition of the theorem. Note that such a condition is readily satisfied (with at least $r^*=0$) when there is no noise in the data (i.e., $v_t=0$ in \eqref{eq:switched-sys} for all $t\in \mathbb{T}$) provided that each subsystem generates enough data.  Moreover, by the second condition of the theorem, exact recovery of the true parameter matrix $A^{\circ}$ is still achievable by the estimator $\Psi$ when $\left\{v_t\right\}$ is a \textit{sparse noise} sequence  containing at most $r^*$ nonzero instances, regardless of the magnitude of these nonzero values. Hence, the larger $r^*$ (i.e., the richer the regression data $\varpi^N$), the more outliers the least absolute deviation LSM estimator can handle.  In contrast, the condition is unlikely to hold generally when \textit{dense noise} is present in the data.

%%%%%%%%%%%%%%%%%%%%%%%%%%%%%%%%%%%%%%%%%%%%%%%%%%%%%%%%%
\section{Error bounds in the presence of noise}\label{sec:Error-Bounds}
As mentioned above, we cannot hope for an exact recovery of the true parameter matrix $A^{\circ}$ by the estimator $\Psi$ from data affected by a dense  noise sequence $\left\{v_t\right\}$.  We need instead to search for  a possible  bound on the estimation error in function of the magnitude of the noise and the richness properties of the data. Indeed \eqref{eq:distance-to-optimal} almost provides such a bound. The remaining question to be investigated is, under which conditions we can lower-bound $\|\phi(\hat{A})-\phi(A^{\circ})\|_1$ by a norm applying directly to $\hat{A}-A^{\circ}$. 
%%%

\subsection{A key step towards the obtention of an error bound}

\noindent To begin with the analysis,  we introduce some useful technical tools, the first of which is the  class of $\mathcal{K}_\infty$ functions (see, e.g., \cite{Kellett14-MCSS}). This class of functions will be used to measure the increasing rate of the estimation error.
%%% 
\begin{definition}[class-$\mathcal{K_\infty}$ functions]
A function $\alpha:\Re_+\rightarrow\Re_+$ is said to be of class-$\mathcal{K_\infty}$ if it is continuous, zero at zero, strictly increasing and satisfies $\lim_{s\rightarrow +\infty }\alpha(s)=+\infty$.
\end{definition}

Using this definition we can state a technical lemma which will play an important role in the analysis. 
%%%%%%%%%%%%%%%%%%%%%%%%%%
\begin{lem}[\cite{Kircher19-TR}]\label{lem:minimum-value}
Let $f:\Re^n\rightarrow\Re_+$ be a positive continuous function satisfying the following properties:
\begin{itemize}
	\item Positive definiteness: $f(x)=0$ if and only if $x=0$
	\item Relaxed homogeneity: There exists a $\mathcal{K}_\infty$ function $q$ such that $f(x)\geq q(\frac{1}{\lambda})f(\lambda x)$ for all $\lambda>0$. 
\end{itemize}
Then for any norm $\left\|\cdot\right\|$ on $\Re^n$,  there exists a constant $\alpha>0$  such that $f(x)\geq \alpha q(\left\|x\right\|)$. 
\end{lem}
%%%%%%%%%%%%%%%%%%%%%%%%%%%%%%%%%%%%%%
Our goal now is to derive a  bound on a certain measure of the parametric estimation error between the true parameter matrix $A^{\circ}$ and the estimated ones $\hat{A}\in \argmin_A\mathcal{J}(A)$. Recalling that $\mathcal{J}(A)$ is invariant under column permutation of the matrix $A$, for this metric to be pertinent, it needs to be specified in terms of distance between the sets $\set(A^{\circ})$  and $\set(\hat{A})$. Hence we consider a metric $d$ of the form $d(A,A')=\left\|A-A'_\pi\right\|$ where $\left\|\cdot\right\|$ is a norm on $\Re^{n\times s}$ and  $\pi:\mathbb{S}\rightarrow\mathbb{S}$ is a permutation depending on the matrices $A$ and $A'$. Here, the notation $A'_\pi$ is used to refer to the matrix obtained by permuting the columns of $A$ as prescribed by $\pi$, $A'_\pi=\bbm a'_{\pi(1)} & \cdots & a'_{\pi(s)}\eem$. The existence of a permutation $\pi$ such that $d(A,A')$ is upper-bounded by $\|\phi(A)-\phi(A')\|_1$  will depend here on the partitions $\left\{I_i(A)\right\}_{i\in \mathbb{S}}$ and $\left\{I_i(A')\right\}_{i\in \mathbb{S}}$ achieved by $A$ and $A'$ respectively on the data set $\varpi^N$.  

\begin{definition}\label{def:comparability}
Consider the data set $\varpi^N$ in \eqref{eq:wN}, generated by the $s$-modes switched system \eqref{eq:switched-sys}. We say that two matrices $A\in \Re^{n\times s}$ and $A'\in\Re^{n\times s}$ are {\it comparable over the data set $\omega^N$} if there exists a permutation $\pi:\mathbb{S}\rightarrow \mathbb{S}$ such that $\left|I_i(A)\cap I_{\pi(i)}(A')\right|\geq \nu_n(X)$ for all $i\in \mathbb{S}$. 
\end{definition}
Note, from Lemma \ref{lem:injectivity} above, that any matrix $A\in \Re^{n\times s}$ such that $\min_{i\in \mathbb{S}}\left|I_i(A)\right|\geq s\nu_n(X)$  is comparable to any other matrix $A'$ with distinct columns satisfying $\phi(A)=\phi(A')$. In that case, it even holds that $A=A'_{\pi}$ for some permutation $\pi$ on $\mathbb{S}$. We state hereafter a sufficient condition for comparability. 
%%%%%%%%%%%%%%%%%%%%%%%%%%%%%%%%%%%%%%%%
\begin{lem}\label{lem:comparability}
Consider a set $\varpi^N$  of input-output data generated by system \eqref{eq:switched-sys} as defined in \eqref{eq:wN}.  Let $A\in \Re^{n\times s}$ be a matrix obeying \eqref{eq:cond-Ni}. 
Then any matrix $A'\in \Re^{n\times s}$ satisfying
\begin{equation}\label{eq:Comparability-Cond}
	\begin{aligned}
		\left|I_i(A)\right|&+\left|I_j(A)\right| \\
	&\geq \max_{\ell\in \mathbb{S}}\big[\left|I_i(A)\cap I_{\ell}(A')\right|+\left|I_j(A)\cap I_{\ell}(A')\right|\big] \\
		 & \qquad \hspace{10pt} +2(s-1)\nu_n(X) \quad   \forall (i,j)\in \mathbb{S}^2, i\neq j,
	\end{aligned}
\end{equation}
 is comparable to $A$ over $\varpi^N$ in the sense of Definition \ref{def:comparability}. 
\end{lem}
\begin{proof}
See Appendix \ref{Proof:lem:comparability}. 
\end{proof}
%%%%%
%%%%%%%%%%%
\noindent To illustrate the condition \eqref{eq:Comparability-Cond}, consider the simple case where $|\mathbb{S}|=s=2$. Then, under the assumption that $A$ is subject to \eqref{eq:cond-Ni}, $A$ and $A'$ are comparable over $\varpi^N$ if
$N\geq \max(|I_1(A')|,|I_2(A')|)+2\nu_n(X)$. 
Noting that $ \max(|I_1(A')|,|I_2(A')|)=N/2+1/2\big||I_1(A')|-|I_2(A')|\big|$ with the outer bars denoting the absolute value, \eqref{eq:Comparability-Cond} reduces to $N\geq 4\nu_n(X)+\big||I_1(A')|-|I_2(A')|\big|$. This relation identifies three factors  which promote comparability: (i) the data $X$ must be generic enough (i.e., $\nu_n(X)$ small); (ii) $A'$ must partition the data into sets of balanced cardinalities; (iii) the number $N$ of data must be large enough.

%%%%%%%%%%%%%%%%%%%%%%%%%%%%%%%%%%%%%%%%%%
%%%%%%%%%%%%%%%%%%%%%%%%%%%
\begin{thm}\label{thm:Error-Bound}
Consider the dataset $\varpi^N$ in \eqref{eq:wN}, generated by the switched system \eqref{eq:switched-sys} and assume that $\xi_r(\varpi^N)<1/2$ for some $r\in \left\{0,\ldots,N\right\}$.   Let $(A,A')\in \Re^{n\times s}\times \Re^{n\times s}$ be a pair of comparable matrices with respect to $\omega^N$ (as defined in Eq. \eqref{eq:wN}). Let $\pi$ denote the associated permutation. Then  for any norm $\left\|\cdot\right\|$ on $\Re^{n\times s}$, there exists a strictly  positive number $D$ such that 
\begin{equation}\label{eq:Error-Bound}
	\|A'_{\pi}-A\|\leq \dfrac{1}{D\big(1-2\xi_r(\varpi^N)\big)}\big(\mathcal{J}(A')-\mathcal{J}(A)+2\delta_r(A)\big). 
\end{equation}
\end{thm}
%%%%%%%%%%%
\begin{proof}
We start by observing that all the conditions of Lemma \ref{lem:INEQ} are satisfied. As a consequence, Eq. \eqref{eq:Bounding-phi(A)-phi(A')} holds. Departing from this equation, we just need to find an appropriate underestimate of $\left\|\phi(A)-\phi(A')\right\|_1$.  To this end, note that 
$$\begin{aligned}
	\left\|\phi(A)-\phi(A')\right\|_1&=\sum_{t\in \mathbb{T}}\big|x_t^\top (a_{\sigma_A(t)}-a'_{\sigma_{A'}(t)})\big|\\
	& = \sum_{(i,j)\in \mathbb{S}^2}\: \: \sum_{t\in I_i(A)\cap I_j(A')}\big|x_t^\top (a_{i}-a'_{j})\big|\\
	%& \geq \sum_{i\in \mathbb{S}}\: \: \sum_{t\in I_i(A)\cap I_{\pi(i)}(A')}\big|x_t^\top (a_{i}-a'_{\pi(i)})\big|
		& \geq \sum_{i\in \mathbb{S}}\: \: \sum_{t\in I_i(A)\cap I_{\pi(i)}(A')}\big|x_t^\top \eta_i\big|
\end{aligned} $$
where  $\eta_i=a_i-a'_{\pi(i)}$ with $\pi:\mathbb{S}\rightarrow\mathbb{S}$ denoting the permutation defining the comparability of $A$ and $A'$ (see Definition \ref{def:comparability}). Recall that $\left|I_i(A)\cap I_{\pi(i)}(A')\right|\geq \nu_n(X)$, $i=1,\ldots,s$. Let $g:\Re^{n\times s}\rightarrow\Re_+$ be the function  defined by  
\begin{equation}\label{eq:g(Lambda)}
	g(\Lambda)=\inf_{\substack{\left\{J_i\right\}_{i\in \mathbb{S}}\\|J_i|\geq \nu_n(X)}} \sum_{i\in \mathbb{S}}\big\|X_{J_i}^\top\eta_i\big\|_1
\end{equation}
where the infimum is taken over all $s$-tuples $(J_1,\ldots,J_s)$ of disjoint  subsets of $\mathbb{T}$ with cardinality larger or equal to $\nu_n(X)$. 
 Then by letting  $\Lambda=A-A'_\pi$, it follows from the inequality above that 
\begin{equation}\label{eq:Ineq-}
	\left\|\phi(A)-\phi(A')\right\|_1\geq g(\Lambda).
\end{equation}
Since the infimum in \eqref{eq:g(Lambda)} operates here on a finite set, it is reached by a certain $(J_1^\star,\ldots,J_s^\star)$. As a consequence $g$ can be expressed by $g(\Lambda)=\sum_{i\in \mathbb{S}}\big\|X_{J_i^\star}^\top\eta_i\big\|_1$. 
\noindent The rest of the proof consists in showing that the function $g$ satisfies the conditions of Lemma \ref{lem:minimum-value}. Clearly, $g$ is positive.  If for some $E=\bbm e_1 & \cdots & e_s\eem \in \Re^{n\times s}$,  $g(E)=0$, then   $X_{J_i^\star}^\top e_i=0$ for all $i=1,\ldots,s$. It follows, by the fact that $|J_i^\star|\geq \nu_n(X)$, that $e_i=0$. Hence  $E=0$ and consequently, $g$ is positive-definite.  Moreover, $g$ is continuous as a consequence of the $\ell_1$ norm being continuous.  
\noindent Finally, $g$ satisfies the relaxed homogeneity property with  the $\mathcal{K}_\infty$ function $q$ defined by $q(x)=x$.  We can therefore apply Lemma \ref{lem:minimum-value} to conclude that $g(\Lambda)\geq D\left\|\Lambda\right\|$ with  $D$ being the strictly positive number defined by
\begin{equation}\label{eq:D-inf}
	D=\inf_{\left\|\Lambda\right\|=1} g(\Lambda). 
\end{equation}
This concludes the proof. 
\end{proof}
%%%%%%%%%%%%%%%%%%
The theorem establishes a bound on the metric $d(A,A')$ in case $A$ and $A'$ are comparable in the sense of Definition \ref{def:comparability}. For a given $r$, it is interesting to note that the bound displayed in \eqref{eq:Error-Bound} is all the smaller as the data are more generic (i.e., $\xi_r(\varpi^N)$ defined in \eqref{eq:nuR} is small for a relatively large $r$). We also note that if $A$ and $A'$ are not comparable as required in the statement of the theorem then, $\left\|A-A'_{\pi}\right\|$ can grow arbitrarily for any permutation $\pi$ while $\left\|\phi(A)-\phi(A')\right\|_1$ remains small. To see this, take for example $s=2$ and 
$$A=\begin{bmatrix}\tilde{a}_1 & \tilde{a}_2\end{bmatrix}, \quad A'=\begin{bmatrix}\tilde{a}'_1 & \beta \tilde{a}'_2\end{bmatrix}$$
with the $\tilde{a}_i$ and $\tilde{a}'_i$ being unit $\ell_2$-norm vectors and $\beta\in \Re$. Then for a given dataset $\varpi^N$ one can choose $\beta$ sufficiently large such that $\sigma_{A'}(t)=1$ for all $t\in \mathbb{T}$, i.e., $I_1(A')=\mathbb{T}$. For such values of $\beta$, $A$ and $A'$ are not comparable in the sense of Definition \ref{def:comparability}. We can see however that $\left\|\phi(A)-\phi(A')\right\|_1$ is independent of $\beta$ while $\left\|A-A'_{\pi}\right\|$ will increase arbitrarily as $\beta$ increases for any permutation $\pi$ on $\mathbb{S}=\left\{1,2\right\}$.  

\begin{rem}
Note that in the scope of Theorem \ref{thm:Error-Bound}, it is, in principle, possible to restrict the defining supremum  of $\xi_r(\varpi^N)$ in  \eqref{eq:nuR} only to all pairs $(A,A')$ of comparable matrices. The interest of such a slight reformulation is that it would produce a smaller value of $\xi_r(\varpi^N)$ and hence a potentially tighter bound in \eqref{eq:Error-Bound}.
\end{rem}

\subsection{Estimation error bound for the switched system}\label{subsec:Err-Bound}
\noindent An interesting situation is when $(A,A')$ is taken in Theorem \ref{thm:Error-Bound} to be equal to $(A^{\circ},\hat{A})$ with $\hat{A}\in \argmin_{A}\mathcal{J}(A)$. In this specific case, invoking the trick used to establish \eqref{eq:distance-to-optimal} yields the following statement. 
%
 
%%%%%%%%%%%%%%%%%%%%%%%%%%%%%%%%%%%%%%
\begin{cor}\label{eq:cor-Error-Boundedness}
Consider the data $\varpi^N$ generated by system \eqref{eq:switched-sys} and assume that $\xi_r(\varpi^N)<1/2$ for some $r\geq 0$. Let   $\hat{A}\in \argmin_{A}\mathcal{J}(A)$. If $\hat{A}$ and the true parameter matrix $A^{\circ}$  are comparable in the sense of Definition \ref{def:comparability} with $\pi:\mathbb{S}\rightarrow \mathbb{S}$ denoting the associated comparability permutation, then for any norm $\left\|\cdot\right\|$ on $\Re^{n\times s}$, there exists a number $D>0$ such that 
\begin{equation}\label{eq:Estimation-Bound}
	\|\hat{A}_{\pi}-A^{\circ}\|\leq \dfrac{2}{D\big(1-2\xi_r(\varpi^N)\big)}\delta_r(A^{\circ}). 
\end{equation}
\end{cor}
%%%%%
Since $r$ can be any integer in $\left\{0,\ldots,N\right\}$ such that $\xi_r(\varpi^N)<1/2$, we can, at least formally, optimize the error bound over all such $r$'s. Hence, whenever the comparability condition holds true,  a better bound can, in principle, be obtained as 
\begin{equation}
	\big\|\hat{A}_{\pi}-A^{\circ}\big\|\leq \min_{r=0,\ldots,N}\Big\{\dfrac{2\delta_r(A^{\circ})}{D\big(1-2\xi_r(\varpi^N)\big)}: \xi_r(\varpi^N)<\dfrac{1}{2}\Big\} 
\end{equation}
As already remarked, $\delta_r(A^{\circ})$ measures how far $\phi(A^{\circ})$ is from the set $\mathcal{S}_r$ of all $r$-sparse signals in $\Re^N$. This is essentially a measure of the amount of noise $\left\{v_t\right\}$ in the system \eqref{eq:switched-sys} which generates the data $\varpi^N$. 
More specifically, $\delta_r(A^{\circ})$ equals the sum of the $N-r$ smallest elements in absolute value of the sequence $\left\{v_t^\circ\right\}_{t\in \mathbb{T}}$ defined by
\begin{equation}\label{eq:vo}
	v_t^\circ = v_t+x_t^\top (a_{\sigma(t)}^{\circ}-a_{\sigma_{A^{\circ}}(t)}^{\circ})
\end{equation}
with $\sigma$ denoting the true switching signal from \eqref{eq:switched-sys}. 
 From the definition of  $\sigma_{A^{\circ}}\in \Sigma$ (see Eq. \eqref{eq:sigmaA}), it is not hard to see that $|v_t^\circ|\leq |v_t|$ for all $t\in \mathbb{T}$ and so, $\delta_r(A^{\circ})\leq \left\|\mathbf{v}\right\|_{1,r} $ with  $\left\|\mathbf{v}\right\|_{1,r}$ denoting the sum, in absolute value, of the $N-r$ smallest entries  of $\left\{v_t\right\}_{t\in \mathbb{T}}$. It follows that under the conditions of Corollary \ref{eq:cor-Error-Boundedness}, 
$\|\hat{A}_{\pi}-A^{\circ}\|\leq 2/\big(D(1-2\xi_r(\varpi^N))\big)\left\|\mathbf{v}\right\|_{1,r}.$    Hence, by considering  the special case where  $r$ is taken equal to $0$ (this is a reasonable choice e.g., when there is no outlier in the data), we get $\|\hat{A}_{\pi}-A^{\circ}\|\leq 2/D\left\|\mathbf{v}\right\|_{1}.$Note that an underestimate $\hat{D}$ of the number $D$ can be numerically found as suggested in Appendix \ref{subsec:estimate-D}.  Using $\hat{D}$ (instead of $D$) in the expression of the bound yields however a more pessimistic value of the bound.  \\
A question we ask now is, under which condition we may have  $v_t^\circ = v_t$ from \eqref{eq:vo}. Such a condition is given in the following proposition. 
%
%%%%%%%%%%%%%%%%%%%%%%%%%%%%%%%%
\begin{prop}\label{prop:distinguishability}
Consider the switched system \eqref{eq:switched-sys} driven by the switching signal $\sigma$ and the noise $\left\{v_t\right\}$. Then a necessary and sufficient condition for $\sigma_{A^{\circ}}= \sigma$ (irrespective of the values of $\sigma$ and those of the noise) is 
\begin{equation}\label{eq:Cond-sigA=sig}
	|v_t|<\dfrac{1}{2}\min_{\substack{(i,j)\in \mathbb{S}^2\\ i\neq j}}\big|x_t^\top (a_i^{\circ}-a_j^{\circ})\big|\: \forall t\in \mathbb{T}.
\end{equation}
\end{prop}
%%%%
\begin{proof}
See Appendix \ref{Proof:prop:distinguishability}. 
\end{proof}
%%%%
%%%
The term on the right hand side of \eqref{eq:Cond-sigA=sig} can be interpreted as a measure of how distinguishable the subsystems are with respect to each other. Hence, what the proposition says is that if the noise level is below a certain threshold (which depends on the parametric distinguishability of the subsystems and on some genericity condition on the regressors), then the true switching signal coincides with $\sigma_{A^{\circ}}$.\\
Finally, an interesting consequence of Proposition \ref{prop:distinguishability} is that, under condition \eqref{eq:Cond-sigA=sig}, we obtain from \eqref{eq:vo} that  $v_t^\circ=v_t$ for all $t\in \mathbb{T}$ with the consequence that $\delta_r(A^{\circ})$ reduces to $\|\mathbf{v}\|_{1,r}$. 

%%%%%%%%%%%%%%%%%%%%%%%%%%%%%

\subsection{On the comparability of $\hat{A}$ and $A^{\circ}$ }
According to Corollary \ref{eq:cor-Error-Boundedness}, a sufficient condition for the estimation error induced by the estimator $\Psi$ to be bounded as in \eqref{eq:Estimation-Bound},  is that of comparability of $\hat{A}$ and $A^{\circ}$ over $\varpi^N$ for all $\hat{A}$ such that $\set(\hat{A})\in \Psi(\varpi^N)$ (see Definition \ref{def:comparability}). Lemma \ref{lem:comparability} suggests that to favor the  comparability of $A^{\circ}$ and $\hat{A}$, the data $\varpi^N$ and the true parameter matrix $A^{\circ}$ should satisfy \eqref{eq:cond-Ni} and \eqref{eq:Comparability-Cond}. Indeed these conditions impose, though in a non trivial way, some constraints on the distinguishability of the modes composing the switched system, the magnitude of the noise, the excitation signal $\left\{u_t\right\}$ and the switching signal $\sigma$. 
%%%%%%%%%%%%%%%

\noindent Intuitively, if the level of the noise $\left\{v_t\right\}$ is low  and if the constituent subsystems are distinguishable enough, then the true parameter matrix $A^{\circ}$ and its estimate $\hat{A}$ should be comparable. We formalize this as follows. 
\begin{lem}\label{lem:Comparability-Estimate}
Assume that the  input-output data $\varpi^N$ \eqref{eq:wN}, generated by the $s$-mode switched system \eqref{eq:switched-sys} is such that $A^{\circ}$ obeys $\min_{i\in \mathbb{S}}\left|I_i(A^\circ)\right|\geq sm$ with $m\geq \nu_n(X)$.   
 Introduce the notation
\begin{equation}\label{eq:gamma}
	\gamma_m=\inf_{\substack{\left\|\eta\right\|_2=1\\|I|\geq m}}\big\|X_{I}^\top \eta\big\|_1,
\end{equation}
where the infimum is taken over all subsets $I$ of $\mathbb{T}$ with cardinality at least $m$ and over all $\eta\in \Re^n$ with unit $\ell_2$ norm. \\
If the subsystems of the switched system \eqref{eq:switched-sys} are parametrically distinguishable enough in the sense that
\begin{equation}\label{eq:distinguishable-2}
	\min_{i\neq j}\left\|a_i^{\circ}-a_j^{\circ}\right\|_2> \dfrac{2\delta_r(A^{\circ})}{\gamma_m\big(1-2\xi_r(\varpi^N)\big)}
\end{equation}
for some $r\in \left\{0,\ldots,N\right\}$ such that $\xi_r(\varpi^N)<1/2$, then $A^{\circ}$ and $\hat{A}$ are comparable over $\varpi^N$ in the sense of Definition \ref{def:comparability} for any $\hat{A}\in \argmin_{A}\mathcal{J}(A)$. 
\end{lem}
\begin{proof}
See Appendix \ref{Proof:lem:Comparability-Estimate}. 
\end{proof}

%%%%%%%%%%%%%%%%%%%%%%%%%%%%%%%%%%
\section{Conclusion}\label{sec:Conclusion}
In this paper we have studied some properties of the least sum-of-minimums (LSM)  absolute deviation estimator for switched system identification. Although this estimator is hard to implement numerically, it serves here as a reference estimator to analyze the degree of richness in the data for the identification scheme to be successful.  In particular, we have proposed a bound on the estimation error induced by this estimator. Interestingly, the expression of the proposed  bound involves explicitly some informativity measures of the training data. The message of that expression in essence is that the richer the data, the smaller the estimation error.  This opens a nice perspective for identification experiment design for switched systems. In effect, one can form an experiment design problem by searching for the input signal which optimizes the derived information-theoretic measures and thereby, the error bound delivered by the estimator. To further pave the avenue towards optimal experiment design, an intermediary step would, perhaps, be to complement the current analysis with one of the LSM estimator when used with the classical quadratic loss. Another important direction of research is to devise efficient numerical routines for estimating the informativity indices derived in this paper. 

%%%%%%%%%%%%%%%%%%%%

\appendix
\section{Proofs}

\subsection{Proof of Lemma \ref{lem:First-Lem}}\label{Proof-lem:First-Lem}
For the sake of notational simplicity we use $\mathcal{T}_r$ in place of $\mathcal{T}_r(v)$.   Let $\mathcal{T}_r^c=\mathbb{T}\setminus \mathcal{T}_r$ be the complement of $\mathcal{T}_r$ in $\mathbb{T}$. 
Then 
 $$
\begin{aligned}
	\left\|v'-v\right\|_1 &= \left\|(v'-v)_{\mathcal{T}_r}\right\|_1+\left\|(v'-v)_{\mathcal{T}_r^c}\right\|_1\\
	& \leq \left\|(v'-v)_{\mathcal{T}_r}\right\|_1+\big\|v'_{\mathcal{T}_r^c}\big\|_1+\left\|v_{\mathcal{T}_r^c}\right\|_1\\
		& = \left\|(v'-v)_{\mathcal{T}_r}\right\|_1+\big\|v'_{\mathcal{T}_r^c}\big\|_1+\inf_{w\in \mathcal{S}_r}\left\|w-v\right\|_1
	\end{aligned}
$$
The inequality is derived from the triangle inequality property of the $\ell_1$ norm. The last equality relation relies on the fact that $\inf_{w\in \mathcal{S}_r}\left\|w-v\right\|_1=\left\|v_{\mathcal{T}_r^c}\right\|_1$ (the sum of the $N-r$ smallest entries in absolute value of $v$). Considering the term $\big\|v'_{\mathcal{T}_r^c}\big\|_1$, we can write
$$ 
\begin{aligned}
	\big\|v'_{\mathcal{T}_r^c}\big\|_1 &=\left\|v'\right\|_1-\left\|v'_{\mathcal{T}_r}\right\|_1\\
	& = \left\|v_{\mathcal{T}_r}\right\|_1-\left\|v'_{\mathcal{T}_r}\right\|_1+\left\|v'\right\|_1
	          -\left(\left\|v\right\|_1-\left\|v_{\mathcal{T}_r^c}\right\|_1\right)\\
	&\leq \left\|(v'-v)_{\mathcal{T}_r}\right\|_1+\left\|v'\right\|_1-\left\|v\right\|_1+\inf_{w\in \mathcal{S}_r}\left\|w-v\right\|_1
\end{aligned}
$$ 
Here, the second equality follows by adding and subtracting $\left\|v_{\mathcal{T}_r}\right\|_1$ while the last line is obtained by applying again the triangle inequality which gives $\big\|v_{\mathcal{T}_r}\big\|_1-\big\|v'_{\mathcal{T}_r}\big\|_1\leq \big\|(v'-v)_{\mathcal{T}_r}\big\|_1$. 
The result follows by combining the second inequality with the first one above.
\qed

\subsection{Proof of Lemma \ref{lem:comparability}}\label{Proof:lem:comparability}
By reasoning as in the proof of Lemma \ref{lem:injectivity} thanks to the fact that $A$ satisfies condition \eqref{eq:cond-Ni}, we reach easily the conclusion that 
for all $i\in \mathbb{S}$, there exists $i^*\in  \mathbb{S}$ such that $\left|I_i(A)\cap I_{i^*}(A')\right|\geq \nu_n(X)$. Let us define a map $\pi: \mathbb{S}\rightarrow  \mathbb{S}$ by posing $\pi(i)=i^*$. We need to show that $\pi$ can be selected to be a permutation under  condition \eqref{eq:Comparability-Cond} of the lemma.   For this purpose, we proceed by contradiction. Recall that $\pi$ is a permutation here if and only if it is injective.  And there is no injective map $\pi$ that  satisfies $\left|I_i(A)\cap I_{\pi(i)}(A')\right|\geq \nu_n(X)$ for all $i\in \mathbb{S}$,  
 if and only if there is a pair $(i,j)$, $i\neq j$, and an index $\ell\in \mathbb{S}$ such that 
\begin{subequations}\label{eq:Contradiction-Cond}
\begin{equation}
	\left\{\begin{aligned}
		&\left|I_i(A)\cap I_{\ell}(A')\right|\geq \nu_n(X) \\
		&\left|I_j(A)\cap I_{\ell}(A')\right|\geq \nu_n(X)
	\end{aligned}\right.
\end{equation}
\mbox{and  $\forall k\neq \ell$,} 
\begin{equation}
	\left\{\begin{aligned}
		&\left|I_i(A)\cap I_{k}(A')\right|< \nu_n(X) \\
		&\left|I_j(A)\cap I_{k}(A')\right|< \nu_n(X)
	\end{aligned}\right. 
	\end{equation}
\end{subequations}
Assume for contradiction that \eqref{eq:Contradiction-Cond} holds. 
Then, because $\left\{I_r(A')\right\}_{r\in \mathbb{S}}$ forms a partition of $\mathbb{T}$, $\left|I_i(A)\right|=\sum_{r=1}^s\left|I_i(A)\cap I_{r}(A')\right|<(s-1)\nu_n(X)+\left|I_i(A)\cap I_{\ell}(A')\right|$. Similarly, we can write, $\left|I_j(A)\right|<(s-1)\nu_n(X)+\left|I_j(A)\cap I_{\ell}(A')\right|$.
Hence $\left|I_i(A)\right|+\left|I_j(A)\right|<2(s-1)\nu_n(X)+\left|I_{i}(A)\cap I_{\ell}(A')\right|+\left|I_{j}(A)\cap I_{\ell}(A')\right|$. This is in contradiction with \eqref{eq:Comparability-Cond}. 
We therefore conclude on the existence of an injective  map (and hence of a permutation) $\pi:\mathbb{S}\rightarrow\mathbb{S}$. 
\qed

\subsection{Proof of Proposition \ref{prop:distinguishability}}\label{Proof:prop:distinguishability}
If \eqref{eq:Cond-sigA=sig} holds true, then  for all $t\in \mathbb{T}$ and all $i\in \mathbb{S}$ with $i\neq \sigma(t)$, 
$$\begin{aligned}
	\big|y_t-x_t^\top a_{\sigma(t)}^{\circ}\big|=|v_t|& <\dfrac{1}{2}\big|x_t^\top(a_{\sigma(t)}^{\circ}-a_i^{\circ})\big|\\
	%& =\dfrac{1}{2}\big|(y_t-x_t^\top a_i^{\circ})-(y_t-x_t^\top a_{\sigma(t)}^{\circ})\big|\\
	& \leq \dfrac{1}{2}\big|y_t-x_t^\top a_i^{\circ}\big|+\dfrac{1}{2}\big|y_t-x_t^\top a_{\sigma(t)}^{\circ}\big|
\end{aligned}$$
where the last inequality is derived from the triangle inequality property of $|\cdot|$. 
 It follows that  $\big|y_t-x_t^\top a_{\sigma(t)}^{\circ}\big|<\big|y_t-x_t^\top a_i^{\circ}\big|$ which implies that $\sigma_{A^{\circ}}(t)=\sigma(t)$ for all $t$.  Conversely, if $\sigma_{A^{\circ}}=\sigma$, then for all $(j,t)\in \mathbb{S}\times \mathbb{T}$ such that $j\neq \sigma(t)$, we get immediately that  $\big|v_t\big|<\big|y_t-x_t^\top a_j^{\circ}\big|=\big|x_t^\top (a_{\sigma(t)}^{\circ}-a_j^{\circ})+v_t\big|$. Taking the square and dividing by $\big|x_t^\top (a_{\sigma(t)}^{\circ}-a_j^{\circ})\big|$ gives $\big|x_t^\top (a_{\sigma(t)}^{\circ}-a_j^{\circ})\big|>-2v_t s_j(t)$ with $ s_j(t)$ denoting the sign of $x_t^\top (a_{\sigma(t)}^{\circ}-a_j^{\circ})$.  The last inequality holds for any possible values of $\sigma$ if and only if $\big|x_t^\top (a_{i}^{\circ}-a_j^{\circ})\big|>2|v_t|$ for all $(i,j)\in \mathbb{S}^2$ with $i\neq j$.  
\qed

\subsection{Proof of Lemma \ref{lem:Comparability-Estimate}}\label{Proof:lem:Comparability-Estimate}
To begin with, let us  observe that by relying on Lemma \ref{lem:minimum-value}, it can be shown that the number $\gamma_m$ in \eqref{eq:gamma} is well defined and satisfies $\gamma_m>0$. 
By the same reasoning as in the proof of Lemma \ref{lem:injectivity}, we know that there  exists a map $\pi:\mathbb{S}\rightarrow\mathbb{S}$ such $\big|I_{i}(A^{\circ})\cap I_{\pi(i)}(\hat{A})\big|\geq m\geq \nu_n(X)$. We just need to establish that such a $\pi$ is bijective under the conditions of the lemma, a property which is equivalent here just to  injectivity of $\pi$. We proceed by contradiction. Suppose that $\pi$ is not injective, that is, we can find $(i,j)\in \mathbb{S}^2$ with $i\neq j$  such that $\pi(i)=\pi(j)$. Let $J_i=I_{i}(A^{\circ})\cap I_{\pi(i)}(\hat{A})$.  By applying Lemma \ref{lem:dist-to-opt}, we can write 
$$\begin{aligned}
	\sum_{i\in \mathbb{S}}\big\|X_{J_i}^\top(a_i^{\circ}-\hat{a}_{\pi(i)})\big\|_1&\leq \big\|\phi(A^{\circ})-\phi(\hat{A})\big\|_1\leq d,
\end{aligned}  $$
where we have posed $d=2\delta_r(A^{\circ})/(1-2\xi_r(\varpi^N))$ for conciseness.  
On the other hand, it follows from the definition \eqref{eq:gamma} of $\gamma_m$ that 
$\sum_{i\in \mathbb{S}}\big\|X_{J_i}^\top(a_i^{\circ}-\hat{a}_{\pi(i)})\big\|_1\geq \gamma_m \sum_{i\in \mathbb{S}}\|a_i^{\circ}-\hat{a}_{\pi(i)}\|_2$. As a consequence, we can write $\sum_{i\in \mathbb{S}}\|a_i^{\circ}-\hat{a}_{\pi(i)}\|_2\leq d/\gamma_m$. Hence, if $\pi(i)=\pi(j)$,  then by virtue of the triangle inequality,  $\|a_i^{\circ}-a_{j}^{\circ}\|_2\leq \|a_i^{\circ}-\hat{a}_{\pi(i)}\|_2+\|a_j^{\circ}-\hat{a}_{\pi(j)}\|_2\leq d/\gamma_m $. This is in contradiction with the assumption \eqref{eq:distinguishable-2}. We therefore conclude that the claim of the lemma holds true.  \qed

\subsection{On the estimation of the number $D$ in \eqref{eq:D-inf}}\label{subsec:estimate-D}
The following lemma provides a method for computing an underestimate of the parameter $D$ in \eqref{eq:D-inf} for a particular choice of the norm involved in its definition though at the price of a combinatorial complexity. 
%%
%\newpage
%
\begin{lem}\label{lem:underestimate-D}
Assume that the norm used for the definition of the number $D$ in \eqref{eq:D-inf} is $\left\|\cdot\right\|_{2,\col}$ defined  by $\left\|\Lambda\right\|_{2,\col}=\sum_{i=1}^s\left\|\eta_i\right\|_2$  for $\Lambda=\big[\begin{matrix}\eta_1 & \cdots& \eta_s\end{matrix}\big]$. Let $m=\nu_n(X)$.
Then 
\begin{equation}\label{eq:Ineq-D-gamma}
	D\geq \gamma_{m}\geq \inf_{|I|= m}\lambda_{\min}^{1/2}(X_IX_I^\top),
\end{equation}
where $\gamma_m$ is the number defined in \eqref{eq:gamma} and $\lambda_{\min}^{1/2}(\cdot)$ denotes the square root of the minimum eigenvalue. The infimum is taken over all subsets $I$ of $\mathbb{T}$ with cardinality equal to $m$. 
\end{lem}
%%
%\newpage
%%%
\begin{proof}
Recall from \eqref{eq:g(Lambda)} and the proof of Theorem \ref{thm:Error-Bound} the expression $g(\Lambda)=\sum_{i\in \mathbb{S}}\big\|X_{J_i^\star}^\top\eta_i\big\|_1$ of the function  $g$, where the $J_i^\star$ are subsets of $\mathbb{T}$ satisfying $|J_i^\star|\geq m=\nu_n(X)$. Then by substituting $\left\|\Lambda\right\|_{2,\col}$ for the norm $\left\|\Lambda\right\|$ in Eq. \eqref{eq:D-inf}, we have 
$$ \begin{aligned}
	D= \inf_{\left\|\Lambda\right\|_{2,\col}=1}g(\Lambda)
	&=\inf_{\left\|\eta_1\right\|_2+\cdots+\left\|\eta_s\right\|_2=1}\sum_{i\in \mathbb{S}}\big\|X_{J_i^\star}^\top\eta_i\big\|_1\\
	& \geq \inf_{\left\|\eta_1\right\|_2+\cdots+\left\|\eta_s\right\|_2=1}\sum_{i\in \mathbb{S}}\gamma_m\|\eta_i\|_2
	 = \gamma_m
\end{aligned}$$
The inequality follows as a consequence of the definition of $\gamma_m$ by which $\big\|X_{J_i^\star}^\top\eta_i\big\|_1\geq \gamma_m \left\|\eta_i\right\|_2$ since $|J_i^\star|\geq m$. Now, to prove the last inequality in \eqref{eq:Ineq-D-gamma}, it suffices to notice that 
 $\left\|X_{I}^\top \eta\right\|_1\geq \left\|X_{I}^\top \eta\right\|_2$. As a result, 
$$\begin{aligned}
	\gamma_m=\inf_{\substack{\left\|\eta\right\|_2=1\\|I|\geq m}}\big\|X_{I}^\top \eta\big\|_1
	&\geq \inf_{\substack{\left\|\eta\right\|_2=1\\|I|\geq m}}\big\|X_{I}^\top \eta\big\|_2\\
	&=\inf_{|I|= m}\lambda_{\min}^{1/2}(X_IX_I^\top).
\end{aligned}
$$  
\end{proof}
\noindent Given  $I\subset \mathbb{T}$, it is easy to obtain $\lambda_{\min}^{1/2}(X_IX_I^\top)$. Hence to obtain an (under)-estimate of $D$, we need to compute  ${N\choose m}$  such values and take the minimum of them. Here the notation ${N\choose m}$ refers to the binomial coefficient.  If we let $\hat{D}=\inf_{|I|= m}\lambda_{\min}^{1/2}(X_IX_I^\top)$, then it follows from \eqref{eq:Estimation-Bound} that 
$\|\hat{A}_{\pi}-A^{\circ}\|\leq \frac{2}{\hat{D}}\|\mathbf{v}\|_1$ in the particular case where $r$ is taken equal to $0$.

%%%%%%%%%%%%%%%%%%%%%%%%%%%%%%%%%%%%%%%%%%%%%%%%%%%%%%%%%%%%%%%%%%%%%%%%%%
\balance
\bibliographystyle{abbrv}
%\bibliography{refs_nips}

\begin{thebibliography}{10}

\bibitem{Bako11-Automatica}
L.~Bako.
\newblock Identification of switched linear systems via sparse optimization.
\newblock {\em Automatica}, 47:668--677, 2011.

\bibitem{Bako17-TAC}
L.~Bako.
\newblock On a class of optimization-based robust estimators.
\newblock {\em IEEE Transactions on Automatic Control}, 62:5990--5997, 2017.

\bibitem{Daubechies10}
I.~Daubechies, R.~DeVore, M.~Fornasier, and C.~S. G\"{u}nt\"{u}rk.
\newblock Iteratively reweighted least squares minimization for sparse
  recovery.
\newblock {\em Communications on Pure and Applied Mathematics}, 63:1--38, 2010.

\bibitem{Garulli12-SYSID}
A.~Garulli, S.~Paoletti, and A.~Vicino.
\newblock A survey on switched and piecewise affine system identification.
\newblock In {\em IFAC Symposium on System Identification, Brussels, Belgium},
  2012.

\bibitem{Goudjil16-ECC}
A.~Goudjil, M.~Pouliquen, E.~Pigeon, and O.~Gehan.
\newblock A real-time identification algorithm for switched linear systems with
  bounded noise.
\newblock In {\em European Control Conference, Alborg, Denmark}, 2016.

\bibitem{Kellett14-MCSS}
C.~M. Kellett.
\newblock A compendium of comparison function results.
\newblock {\em Mathematics of Control, Signals, and Systems}, 26:339--374,
  2014.

\bibitem{Kircher19-TR}
A.~Kircher, L.~Bako, E.~Blanco, and M.~Benallouch.
\newblock An optimization framework for resilient batch estimation in
  cyber-physical systems.
\newblock Technical report, Ecole Centrale de Lyon (arxiv.org/abs/1906.01714),
  2019.

\bibitem{Lauer18-Automatica}
F.~Lauer.
\newblock Global optimization for low-dimensional switching linear regression
  and bounded-error estimation.
\newblock {\em Automatica}, 89:73--82, 2018.

\bibitem{Lauer19-Book}
F.~Lauer and G.~Bloch.
\newblock {\em Hybrid System Identification: Theory and Algorithms for Learning
  Switching Models}.
\newblock Springer International Publishing, 2019.

\bibitem{Liberzon03-Book}
D.~Liberzon.
\newblock {\em Switching in Systems and Control}.
\newblock Birkhauser Boston Inc., 2003.

\bibitem{Lunze-Lamnabhi2009-Book}
J.~Lunze and F.~{Lamnabhi-Lagarrigue (Eds)}.
\newblock {\em Handbook of Hybrid Systems Control: Theory, Tools,
  Applications}.
\newblock Cambridge University Press, 2009.

\bibitem{Ozay12-TAC}
N.~Ozay, M.~Sznaier, C.~Lagoa, and O.~Camps.
\newblock A sparsification approach to set membership identification of
  switched affine systems.
\newblock {\em IEEE Transactions on Automatic Control}, 57:634--648, 2012.

\bibitem{Paoletti07}
S.~Paoletti, A.~Juloski, G.~Ferrari-Trecate, and R.~Vidal.
\newblock Identification of hybrid systems: A tutorial.
\newblock {\em European Journal of Control}, 13:242--260, 2007.

\bibitem{Petreczky11-CDC}
M.~Petreczky and L.~Bako.
\newblock On the notion of persistence of excitation for linear switched
  systems.
\newblock In {\em IEEE Conference on Decision and Control and European Control
  Conference, Orlando, FL, USA}, 2011.

\bibitem{Petreczky20-IJNRC}
M.~Petreczky, L.~Bako, S.~Lecoeuche, and K.~Motchon.
\newblock Minimality and identifiability of discrete-time {SARX} systems.
\newblock {\em To appear in International Journal of Robust and Nonlinear
  Control}, 2020.

\bibitem{Pillonetto16-Automatica}
G.~Pillonetto.
\newblock A new kernel-based approach to hybrid system identification.
\newblock {\em Automatica}, 70:21--31, 2016.

\bibitem{Sun05-Book}
Z.~Sun.
\newblock {\em Switched Linear Systems: Control and Design}.
\newblock Springer-Verlag London, 2005.

\bibitem{Vidal08-Automatica}
R.~Vidal.
\newblock Recursive identification of switched {ARX} systems.
\newblock {\em Automatica}, 44:2274--2287, 2008.

\bibitem{Vidal03}
R.~Vidal, S.~Soatto, Y.~Ma, and S.~Sastry.
\newblock An algebraic geometric approach to the identification of a class of
  linear hybrid systems.
\newblock In {\em Conference on Decision and Control, Maui, Hawaii, USA}, 2003.

\end{thebibliography}

\end{document}